\documentclass[runningheads]{llncs}
\usepackage[T1]{fontenc}
\usepackage{graphicx}
\usepackage{amsmath,amssymb,amsfonts}%
 \usepackage{MnSymbol}
\usepackage{xcolor}%

\begin{document}
\title{Torsion of $\alpha$-connections on the density manifold}
%
%
\author{Nihat Ay \inst{1,2,3}\orcidID{0000-0002-8527-2579} \and
Lorenz J. Schwachh\"ofer \inst{4}\orcidID{0000-0002-4268-6923}}

\authorrunning{N. Ay and L.J. Schwachh\"ofer}
%
\institute{Institute for Data Science Foundations, Hamburg University of Technology, Blohmstra{\ss}e 15, 21073 Hamburg, Germany
\email{nihat.ay@tuhh.de}
\and
Santa Fe Institute, 1399 Hyde Park Road, Santa Fe, NM 87501, USA
\and Faculty of Mathematics and Computer Science, Leipzig University, Augustusplatz 10, 04109 Leipzig, Germany
\\
\and Department of Mathematics, TU Dortmund University, Vogelpothsweg 87, 44221 Dortmund, Germany
\email{lschwach@math.tu-dortmund.de}}
\maketitle              
\begin{abstract}
We study the torsion of the $\alpha$-connections defined on the density manifold in terms of a regular Riemannian metric. 
In the case of the Fisher-Rao metric our results confirm the fact that all  $\alpha$-connections are torsion free. For the $\alpha$-connections obtained by the 
Otto metric, we show that, except for $\alpha = -1$, they are not torsion free. 

\end{abstract}
\section{Introduction}
Classical information-geometric structures are agnostic to the geometry of the sample space $M$ which is given, for instance, in terms of a Riemannian metric 
${\rm g}$. To address this problem, many works within information geometry have tried to incorporate Wasserstein geometry \cite{KZ22}. In this context, the so-called Otto metric has been studied as an alternative to the Fisher-Rao metric \cite{Otto}.
This line of research has been recently extended by the definition of an exponential connection as the dual of the mixture connection with respect to the Otto metric \cite{Ay}. It was conjectured that the new exponential connection is torsion free, which would directly imply counterparts of important results within the classical theory, such as the generalised Pythagorean theorem in terms of a canonical contrast function.

In the present article, we disprove this conjecture by deriving an explicit formula for the torsion of the new exponential connection. In fact, we study the torsion of all $\alpha$-connections associated with some regular Riemannian metric on the density manifold. For $\alpha = 0$, we do not recover the Levi-Civita connection whenever the connection for $\alpha = 1$ is not torsion free.

\section{The affine structure of $P^\infty_+(M)$}

We consider a finite-dimensional compact connected $C^\infty$-manifold $M$ without boundary. A {\em smooth measure on $M$ }is a measure which is induced by some smooth Riemannian metric on $M$. By compactness of $M$, such a measure is finite, and we are particularly interested in the space $P_+^\infty(M)$ of smooth probability measures.

Picking some $\mu \in P_+^\infty(M)$, we may identify the {\em vector space $S^\infty(M)$ of smooth signed measures on $M$} with the space $C^\infty(M$) of arbitrarily often differentiable real valued functions on $M$. This is done by the linear isomorphism
\begin{equation} \label{define-imath-mu}
\imath_{\mu}: C^\infty(M) \longrightarrow S^\infty(M), \qquad f \longmapsto f\; \mu,
\end{equation}
whose inverse is given by the Radon-Nikodym derivative w.r.t. $\mu$. We also consider the codimension one subspaces
\begin{eqnarray} \label{S0-infty}
      {S}^\infty_0(M) \; & := & \; \left\{ \nu \in S^\infty(M) \; : \;  \int_M \nu = 0 \right\},\\
      \nonumber C^\infty_0(M, \mu) \; & := & \; \left\{ f \in C^\infty(M) \; : \; \int_M f\; d\mu = 0 \right\},
\end{eqnarray}
so that obviously, $\imath_{\mu}$ from (\ref{define-imath-mu}) restricts to an isomorphism of $C^\infty_0(M, \mu)$ and $S^\infty_0(M)$. Morover, $P^\infty_+(M)$ is the image of the intersection of the affine hyperplane $1 + C^\infty_0(M, \mu)$ with the convex open set of positive functions.

As $C^\infty(M)$ has a canonical Fr\'echet space structure \cite{KM97}, it follows that $\imath_{\mu}$ induces a Fr\'echet space structure on $S^\infty(M)$ and $S^\infty_0(M)$, and $P^\infty_+(M)$ has the structure of an open convex set in an affine hyperplane. It is straightforward to verify that these Fr\'echet structures are independent of the choice of $\mu$. Therefore, there is a canonical identification
\begin{equation} \label{tangent-space}
T_\mu P^\infty_+(M) \cong S^\infty_0(M), \qquad TP^\infty_+(M) \cong P^\infty_+(M) \times S^\infty_0(M),
\end{equation}
and the canonical flat connection on the vector space $S^\infty(M)$ restricts to a canonical flat affine connection $\nabla^{({\rm m})}$ on $P^\infty_+(M)$, called the {\em mixture connection}. It is characterized by the fact that the constant vector fields $(A_\mu) := (\mu, \nu_0)$ are $\nabla^{({\rm m})}$-parallel for each $\nu_0 \in S^\infty_0(M)$. In particular, the $\nabla^{({\rm m})}$-geodesics are the straight line segments, and due to the convexity, any two elements in $P^\infty_+(M)$ can be joined by a unique $\nabla^{({\rm m})}$-geodesic.

\section{Regular metrics on $P^\infty_+(M)$}

A {\em Riemannian metric on $P^\infty_+(M)$ }is a collection of inner products ${\mathcal G} =({\mathcal G}_\mu)_{\mu \in P^\infty_+(M)}$ on each $T_\mu P^\infty_+(M) \cong S^\infty_0(M)$ which depends smoothly on $\mu$, in the sense that smooth vector fields $A, B$ on $P^\infty_+(M)$ (i.e., smooth families $A = (A_\mu)_{\mu \in P^\infty_+(M)}$, $B = (B_\mu)_{\mu \in P^\infty_+(M)} \in S^\infty_0(M)$) yield a smooth function on $P^\infty_+(M)$
\[
\mu \longmapsto {\mathcal G}_\mu(A_\mu, B_\mu).
\]

\

\noindent {\bf Example: The Fisher-Rao metric}

Given $\mu \in P^\infty_+(M)$, we use (\ref{define-imath-mu}) to identify $T_\mu P^\infty_+(M) \cong S^\infty_0(M) \cong C^\infty_0(M, \mu) \subset L^2(M, \mu)$, where the latter is the Hilbert space of square integrable functions. The {\em Fisher-Rao metric on $P^\infty_+(M)$ }is then the restricton of the inner product on $L^2(M, \mu)$ to $C^\infty_0(M, \mu)$ under this identification. That is,
\begin{equation} \label{FR}
\llangle A, B \rrangle^{\rm FR}_\mu := \int_M \frac{dA_\mu}{d\mu} \; \frac{dB_\mu}{d\mu}\; d\mu.
\end{equation}

\

The Fisher-Rao metric has very natural properties; e.g. it is invariant under diffeomorphisms of $M$ and is in fact (up to constant rescaling) the only Riemannian metric on $P^\infty_+(M)$ with this property \cite{BBM16}.

However, there are other important choices of such metrics, for instance the {\em Otto metric }which we shall discuss in more detail in section \ref{section-Otto}. For the moment, we shall stay fairly general in our discussion, but always keep the Otto metric in mind as our main application.

\begin{definition} \label{def-regular} A {\em regular Riemannian metric on $P^\infty_+(M)$ }is a metric of the form
\begin{equation} \label{regular-metric}
{\mathcal G}_\mu(A, B) := \llangle \Phi^{\mathcal G}_\mu (A_\mu), B_\mu \rrangle^{\rm FR}_\mu,
\end{equation}
where $\Phi^{\mathcal G}_\mu: S^\infty_0(M) \to S^\infty_0(M)$ is an automorphism of Fr\'echet spaces.
\end{definition}

Evidently, this is only well defined if for each $\nu \in S^\infty_0(M)$ the expression $(\Phi^{\mathcal G}_\mu)(\nu)$ depends smoothly on $\mu$, and if
\begin{equation} \label{eq:selfadjoint}
\llangle \Phi^{\mathcal G}_\mu (\nu^1), \nu^2 \rrangle^{\rm FR}_\mu = \llangle \Phi^{\mathcal G}_\mu (\nu^2), \nu^1 \rrangle^{\rm FR}_\mu, \qquad \llangle \Phi^{\mathcal G}_\mu (\nu), \nu \rrangle^{\rm FR}_\mu > 0, \; \nu \neq 0.
\end{equation}

We also shall use the notation for the G\^ateaux derivative
\begin{equation} \label{direction-deriv}
\partial_{\nu^1} \Phi^{\mathcal G}_\mu (\nu^2) := \left. \frac d{dt} \right|_{t=0} \Phi^{\mathcal G}_{\mu + t \nu^1} (\nu^2) \in S^\infty_0(M).
\end{equation}

\section{The $\alpha$-connections of a regular metric on $P^\infty_+(M)$}

For a Riemannian metric ${\mathcal G}$ on $P^\infty_+(M)$ its {\em covariant mixture derivative $\nabla^{({\rm m})} {\mathcal G}$} is defined as the tensor
\begin{eqnarray} \label{Amari-Tensor}
{\mathcal A}^{\mathcal G}(A; B, C) \; & := & \; - \left( \nabla^{({\rm m})}_A {\mathcal G} \right) (B, C)\\
\nonumber \; & = & \; - A \left( {\mathcal G}(B, C) \right) + {\mathcal G}\left( \nabla^{({\rm m})}_A B, C \right) + {\mathcal G} \left( B, \nabla^{({\rm m})}_A C \right).
\end{eqnarray}

It is straightforward to verify that this expression is tensorial in all entries, so that ${\mathcal A}^{\mathcal G}(A; B, C)_\mu$ only depends on $A_\mu, B_\mu$ and $C_\mu$. Furthermore, it is obviously symmetric in the last two entries.

\begin{lemma}
We have the identity
\begin{equation} \label{define-K}
{\mathcal A}^{\mathcal G}(A; B, C) = {\mathcal G}\left({\mathcal K}^{\mathcal G}(A, B), C\right),
\end{equation}
where
\begin{equation} \label{define-K2}
{\mathcal K}^{\mathcal G}(A, B)_\mu = (\Phi^{\mathcal G}_\mu)^{-1} \left(\frac{d{A_\mu}}{d\mu} \; \Phi^{\mathcal G}_\mu (B_\mu) - \partial_{A_\mu} \Phi^{\mathcal G}_\mu(B_\mu)\right),
\end{equation}
using the notation (\ref{direction-deriv}).
\end{lemma}

\begin{proof}
Let $\mu \in P^\infty_+(M)$ and $a, b, c \in C^\infty_0(M, \mu)$, so that the constant vector fields $A : \equiv a \mu$, $B :\equiv b\mu$ and $C :\equiv c \mu$ on $P^\infty_+(M)$ are $\nabla^{({\rm m})}$-parallel, i.e., $\nabla^{({\rm m})}_A B = \nabla^{({\rm m})}_A C = 0$. Then (\ref{Amari-Tensor}) implies
\begin{eqnarray*}
{\mathcal A}^{\mathcal G}(A; B, C) \; & = & \; -A \llangle B, C \rrangle^{\mathcal G}_\mu = -\left. \frac d{dt} \right|_{t=0} \llangle \Phi^{\mathcal G}_{(1+ta)\mu} (B), C \rrangle^{\rm FR}_{(1+ta)\mu}\\
& = & \; -\left. \frac d{dt} \right|_{t=0} \int_M \frac{d \Phi^{\mathcal G}_{(1 + t a)\mu} (B)}{d((1 + t a) \mu)} \frac{dC}{d((1 + ta) \mu)} \; d((1 + ta) \mu)\\
& = & \; \int_M \left( a \; \frac{d \Phi^{\mathcal G}_\mu (B)}{d\mu} - \frac{d(\partial_{a \mu}\Phi^{\mathcal G}_\mu(B))}{d\mu} \right)  \frac{dC}{d\mu} \; d\mu\\
& = & \; \left\llangle \left(\frac{dA_\mu}{d\mu} \; \Phi^{\mathcal G}_\mu (B_\mu) - \partial_{A_\mu} \Phi^{\mathcal G}_\mu(B_\mu)\right) , C_\mu 
\right\rrangle_\mu^{\rm FR},
\end{eqnarray*}
which together with (\ref{regular-metric}) implies (\ref{define-K2}).
\end{proof}

\begin{definition} Let $\mathcal G$ be a regular metric on $P^\infty_+(M)$, given by (\ref{regular-metric}) for some map $\phi_\mu^{\mathcal G}: S^\infty_0(M) \to S^\infty_0(M)$. Then for $\alpha \in {\mathbb R}$ we define the {\em $\alpha$-connection of $\mathcal G$ }by
\begin{equation} \label{alpha-connection}
\nabla^{({\mathcal G}, \alpha)}_A B := \nabla^{({\rm m})}_A B - \frac{\alpha+1}2 {\mathcal K}^{\mathcal G}(A, B).
\end{equation}
\end{definition}

Observe that according to (\ref{alpha-connection}) we have ${\mathcal K}^{\mathcal G}(A, B) = \nabla^{({\rm m})}_A B - \nabla^{({\mathcal G}, 1)}_A B$, whence (\ref{alpha-connection}) is equivalent to
\[
\nabla^{({\mathcal G}, \alpha)}_A B \; = \; \frac{1 - \alpha}2\; \nabla^{({\rm m})}_A B + \frac{1 + \alpha}2\; \nabla^{({\mathcal G}, 1)}_A B.
\]
We now collect the following properties of these $\alpha$-connections.

\begin{theorem}
\begin{enumerate}
\item Equation (\ref{alpha-connection}) defines indeed a connection of $P^\infty_+(M)$.
\item The torsion of $\nabla^{({\mathcal G}, \alpha)}$ is given by
\begin{equation} \label{Tor-alpha}
{\rm Tor}^{({\mathcal G}, \alpha)}(A, B)\; = \; \frac{\alpha+1}2\; {\rm Tor}^{({\mathcal G}, 1)}(A, B)\; = \; \frac{\alpha+1}2\; \big({\mathcal K}^{\mathcal G}(B, A) - {\mathcal K}^{\mathcal G}(A, B)\big).
\end{equation}
That is, apart from $\nabla^{(-1)} = \nabla^{({\rm m})}$, either {\em all }or {\em none }of the $\nabla^{(\alpha)}$, $\alpha \neq -1$, have torsion.
\item $\nabla^{({\mathcal G}, \alpha)}$ and $\nabla^{({\mathcal G}, -\alpha)}$ are conjugate w.r.t. ${\mathcal G}$, i.e., we have the identity
\[
A\; {\mathcal G}(B, C) = {\mathcal G}\left(\nabla^{({\mathcal G}, \alpha)}_A B, C\right) + {\mathcal G}\left( B, \nabla^{({\mathcal G}, -\alpha)}_A C \right)
\]
for all vector fields $A, B, C$ on $P^\infty_+(M)$. In particular, $\nabla^{({\mathcal G}, 0)}$ is compatible with ${\mathcal G}$, and the connection $\nabla^{({\mathcal G}, 1)}$ has vanishing curvature.
\item The Levi-Civita connection of ${\mathcal G}$ is given as
\begin{equation} \label{LC}
\nabla^{LC}_A B = \nabla^{({\rm m})}_A B - \frac12 {\mathcal K}^{\mathcal G}(A, B) - \frac12 {\mathcal K}^{\mathcal G}(B, A) + \frac12 {\mathcal D}^{\mathcal G}(A, B),
\end{equation}
where ${\mathcal D}^{\mathcal G}$ is the symmetric tensor determined by the identity
\begin{equation} \label{DG-tensor}
{\mathcal G}({\mathcal D^{\mathcal G}}(A, B), C) \; = \; {\mathcal A}^{\mathcal G}(C; A, B).
\end{equation}
\end{enumerate}
\end{theorem}

\begin{proof} Since ${\mathcal A}^{\mathcal G}$ is tensorial, (\ref{define-K}) immediately implies that ${\mathcal K}^{\mathcal G}$ is tensorial in both entries as well. From this, it is straightforward to verify the connection properties for $\nabla^{({\mathcal G}, \alpha)}$.

The remaining statements follow by straightforward calculations, cf. \cite[ch. 3]{AN06}. In particular, $\nabla^{({\mathcal G}, 1)}$ is conjugate to $\nabla^{({\mathcal G}, -1)} = \nabla^{({\rm m})}$ w.r.t. ${\mathcal G}$ and $\nabla^{({\rm m})}$ is flat. Being conjugate to a flat connection, $\nabla^{({\mathcal G}, 1)}$ has vanishing curvature as well.
\end{proof}

\begin{corollary} \label{Geodesics}
For a regular metric ${\mathcal G}$ on $P^\infty_+(M)$ the following are equivalent.
\begin{enumerate}
\item $\nabla^{({\mathcal G}, \alpha)}$ is torsion free for all $\alpha \in {\mathbb R}$.
\item $\nabla^{({\mathcal G}, 0)}$ coincides with the Levi-Civita connection $\nabla^{LC}$ of ${\mathcal G}$.
\item $\nabla^{({\mathcal G}, 0)}$ and $\nabla^{LC}$ have the same geodesics.
\end{enumerate}
\end{corollary}

\begin{proof} As $\nabla^{({\mathcal G}, 0)}$ is compatible with ${\mathcal G}$, it coincides with $\nabla^{LC}$ iff $\nabla^{({\mathcal G}, 0)}$ is torsion free, which by (\ref{Tor-alpha}) is equivalent to the first statement. So the first two statements are equivalent, and second statement obviously implies the third. For the converse, let $\gamma_t$ be a common geodesic of $\nabla^{({\mathcal G}, 0)}$ and $\nabla^{LC}$. Then
\[
0 = (\nabla^{LC}_{\dot \gamma_t} \dot \gamma_t) - (\nabla^{({\mathcal G}, 0)}_{\dot \gamma_t} \dot \gamma_t) \stackrel{(\ref{LC})} = - \frac12 {\mathcal K}^{\mathcal G}(\dot \gamma_t, \dot \gamma_t) + \frac12 {\mathcal D}^{\mathcal G}(\dot \gamma_t, \dot \gamma_t).
\]
Thus for any tangent vector $C$,
\[
0 = {\mathcal G}({\mathcal K}^{\mathcal G}(\dot \gamma_t, \dot \gamma_t), C) - {\mathcal G}({\mathcal D}^{\mathcal G}(\dot \gamma_t, \dot \gamma_t), C) \stackrel{(\ref{Tor-alpha}), (\ref{DG-tensor})} = {\mathcal A}^{\mathcal G}(\dot \gamma_t; \dot \gamma_t, C) - {\mathcal A}^{\mathcal G}(C; \dot \gamma_t, \dot \gamma_t).
\]
Since this holds for {\em all }geodesics and hence for {\em all tangent vectors $\dot \gamma_t$}, polarization implies 
\[
0 = {\mathcal A}^{\mathcal G}(A; B, C) + {\mathcal A}^{\mathcal G}(B; A, C) - 2 {\mathcal A}^{\mathcal G}(C; A, B).
\]
From this, standard tensor calculations imply that ${\mathcal A}^{\mathcal G}$ is totally symmetric and hence, ${\mathcal K}^{\mathcal G}$ is symmetric by (\ref{define-K}), so that (\ref{Tor-alpha}) implies the first statement.
\end{proof}

\noindent {\bf Example: The Fisher-Rao metric}

According to (\ref{regular-metric}) the Fisher-Rao metric is a regular metric with $\Phi^{\rm FR}_\mu = {\mathbf 1}$ for all $\mu$. Thus, it follows that $\partial_A \Phi^{\rm FR}_\mu = 0$, whence letting $A_\mu := a\mu$, $B_\mu := b\mu$ and $C_\mu := c \mu$ we have according to (\ref{define-K2})
\begin{equation} \label{K-FR}
{\mathcal K}^{\rm FR}_\mu(a \mu, b \mu) = \frac{dA_\mu}{d\mu} \; B_\mu = ab \; \mu = {\mathcal K}^{\rm FR}_\mu(b\mu, a\mu),
\end{equation}
so that by (\ref{Amari-Tensor}),
\[
{\mathcal A}^{\rm FR}_\mu(a\mu, b\mu, c\mu) = \llangle ab \; \mu, c\; \mu \rrangle_\mu^{\rm FR} \; = \; \int_M abc\;{d\mu} .
\]
Thus, ${\mathcal A}^{\rm FR}$ coincides with the classical {\em Amari-Chentsov tensor }on $P^\infty_+(M)$.

That is, the $\alpha$-connections $\nabla^{({\rm FR}, \alpha)}$ from (\ref{alpha-connection}) for the Fisher-Rao metric coincide with the standard $\alpha$-connections defined e.g. in \cite[ch. 3]{AN06} or \cite{AJLS}. They are torsion free as ${\mathcal K}^{\rm FR}$ is symmetric by (\ref{K-FR}). In particular, $\nabla^{({\rm FR},0)} = \nabla^{\rm LC}$ is the Levi-Civita-Connection of the Fisher-Rao metric, and $\nabla^{({\rm FR},1)} =: \nabla^{({\rm e})}$ is the {\em exponential connection}, i.e., the flat connection conjugate to $\nabla^{({\rm m})}$ \cite{AN06}, \cite{AJLS}.

Thus, the $\alpha$-connections from (\ref{alpha-connection}) may be regarded as a canonical generalization of the $\alpha$-connections of the Fisher-Rao metric to arbitrary regular metrics.

\section{The Otto metric} \label{section-Otto}

We shall now present another important regular Riemannian metric on $P^\infty_+(M)$, called the {\em Otto metric}, as it was first introduced by F. Otto in \cite{Otto}. The dual connection of $\nabla^{({\rm m})}$ w.r.t. this metric was considered by the first named author in \cite{Ay}, where it was denoted by $\nabla^{({\rm e}_1)}$.

Let $(M, {\rm g})$ be a compact Riemannian manifold without boundary. We shall use the notation ${\rm g}(\xi, \zeta) =: \langle \xi, \zeta \rangle$ for $\xi, \zeta \in T_xM$. Let ${\mathcal T}(M)$ denote the Lie algebra of vector fields on $M$. Given $f \in C^\infty(M)$, the {\em gradient vector field of $f$ (w.r.t. ${\rm g}$) }is the vector field ${\rm grad}(f) \in {\mathcal T}(M)$ determined by the property
\begin{equation} \label{grad-g0}
\langle {\rm grad} \, f, X \rangle = X(f) \qquad \text{for any vector field $X \in \mathcal{T}(M)$}.
\end{equation}

In order to construct the Otto metric, we need some properties of the {\em divergence of vector fields}, and the {\em Laplace operator}, of which we only collect the most important facts; for a more comprehensive explanation, cf. \cite{Ay}.

\begin{proposition} Let $(M, {\rm g})$ be as above and $\mu \in P^\infty_+(M)$.
\begin{enumerate}
\item For a vector field $X \in {\mathcal T}(X)$ there is a unique function ${\rm div}_\mu(X) \in C^\infty_0(M, \mu)$, called the {\em divergence of $X$ w.r.t. $\mu$}, satisfying
\begin{equation} \label{divergence}
\int_M X(f)\; d\mu = - \int_M f \; {\rm div}_\mu(X)\; d\mu \qquad \text{for all $f \in C^\infty(M)$}.
\end{equation}
\item For $f \in C^\infty(M)$ define the {\em $\mu$-Laplacian of $f$ }as
\begin{equation} \label{Laplace}
\Delta_\mu f := {\rm div}_\mu ({\rm grad} f).
\end{equation}
Then $\Delta_\mu f \in C^\infty_0(M, \mu)$ and $\ker \Delta_\mu$ consists of the constant functions.
\item The restriction $\Delta_\mu: C^\infty_0(M, \mu) \to C^\infty_0(M, \mu)$ is an isomorphism of Fr\'echet spaces.
\end{enumerate}
\end{proposition}

It follows that
\begin{equation} \label{Laplace-symm}
-\int_M (\Delta_\mu f)\; g\; d\mu = \int_M ({\rm grad} f)(g)\; d\mu = \int_M \langle {\rm grad} f, {\rm grad} g \rangle\; d\mu.
\end{equation}

With this notation, we define the {\em Otto metric on $P^\infty_+(M)$ }by
\begin{equation} \label{Ottometric}
\llangle a \mu, b \mu \rrangle^{\rm O}_\mu := \int_M \left\langle {\rm grad} (\Delta_\mu^{-1}a ),{\rm grad}\left(\Delta_\mu^{-1} b \right) \right \rangle\; d\mu.
\end{equation}

We assert that this is a regular metric in the sense of Definition \ref{def-regular}. Bilinearity, symmetry are nonnegativity are immediate. Also, note that $\llangle \nu, \nu \rrangle_\mu^{\rm O} = 0$ iff $\nu = c \mu$ for some constant $c \in {\mathbb R}$. But then, $\nu \in S^\infty_0(M)$ implies $c = 0$, whence $\nu = 0$. That is, the Otto metric defines an inner product on each tangent space.

Moreover, comparing (\ref{Laplace-symm}) and ({\ref{Ottometric}), we can rewrite
\begin{equation} \label{Otto}
\llangle a\mu, b \mu \rrangle_\mu^{\rm O} = - \int_M (\Delta_\mu^{-1} a )\; b\; d\mu = \left\llangle -(\Delta_\mu^{-1} a)\; \mu, b \mu \right\rrangle_\mu^{\rm FR}.
\end{equation}
Thus, the Otto metric can be written in the form (\ref{regular-metric}) with
\begin{equation} \label{Phi-O}
\Phi^{\rm O}_\mu (a \mu) := - (\Delta_\mu^{-1} a)\; \mu \qquad \Longrightarrow \qquad \Phi^{\rm O}_\mu = - \imath_\mu \circ \Delta_\mu^{-1} \circ \imath_\mu^{-1}
\end{equation}
with $\imath_\mu$ from (\ref{define-imath-mu}), so that $\Phi^{\rm O}_\mu$ is the composition of Fr\'echet space isomorphisms, showing its regularity. Then the tensor fields ${\mathcal K}^{\rm O}$ from (\ref{define-K2}), ${\mathcal A}^{\rm O}$ from (\ref{define-K}), the torsion of $\nabla^{({\rm O}, \alpha)}$ and the Levi-Civita connection of $\llangle \cdot, \cdot \rrangle^{\rm O}$ are as follows.

\begin{theorem}
Let $(M, {\rm g})$ as before and $\mu \in P^\infty_+(M)$. Then for $a \mu, b \mu \in S^\infty_0(M) \cong T_\mu P^\infty_+(M)$ and $A = a\mu, B = b\mu$ we have
\begin{eqnarray} \label{KO}
{\mathcal K}^{\rm O}_\mu (a \mu, b \mu) \; & = & \; {\mathcal K}^{\rm FR}_\mu(a\mu, b\mu) + \langle {\rm grad}\; (\Delta_\mu^{-1} b), {\rm grad}\; a \rangle\; \mu.\\
\label{AO}
{\mathcal A}^{\rm O}_\mu(a\mu, b\mu, c\mu) \; & = & \; -\left\llangle a \mu, \; \Delta_\mu \langle {\rm grad} (\Delta^{-1}_\mu b), {\rm grad}(\Delta_\mu^{-1} c) \rangle \; \mu \right \rrangle^{\rm O}.\\
\nonumber
{\rm Tor}^{({\rm O}, \alpha)}_\mu(a \mu, b \mu) \; & = & \; \frac{\alpha+1}2 \Big(\langle {\rm grad}\; (\Delta_\mu^{-1} a), {\rm grad}\; b \rangle\\
\label{TorO}
&& \qquad \qquad \qquad \qquad - \langle {\rm grad}\; (\Delta_\mu^{-1} b), {\rm grad}\; a \rangle\Big)\; \mu.\\
\nonumber
\nabla^{LC}_A B \; &= & \; \nabla^{({\rm m})}_A B - \frac12 \Big(2 ab \; + \;  \Delta_\mu \langle {\rm grad} (\Delta^{-1}_\mu a), {\rm grad}(\Delta_\mu^{-1} b) \rangle\\
\label{LC-O}
&& \;\;\langle {\rm grad}\; (\Delta_\mu^{-1} a), {\rm grad}\; b \rangle\; + \;\langle {\rm grad}\; a, {\rm grad}\; (\Delta_\mu^{-1} b) \rangle\Big) \; \mu.
\end{eqnarray}
\end{theorem}

\begin{proof} For the proof, we shall use the identity \cite[Proposition 1, (8)]{Ay}
\begin{equation} \label{divergence-mu}
{\rm div}_\mu (X) = {\rm div}_{\mu_{\rm g}} (X) + \left\langle X, {\rm grad}\; \left(\ln \frac{d\mu}{d\mu_{\rm g}} \right)\right\rangle
\end{equation}
which -- substituting $X = {\rm grad}\; h$ for any $h \in C^\infty(M)$ -- yields
\begin{equation} \label{Delta-mu}
\Delta_\mu h = \Delta_{\mu_{\rm g}} h + \left\langle {\rm grad}\; h, {\rm grad}\; \left(\ln \frac{d\mu}{d\mu_{\rm g}} \right) \right\rangle.
\end{equation} 

Let $\mu := \rho \mu_{\rm g}$ and $\mu_t :=  (1 + t a) \mu = (1 + t a) \rho \mu_{\rm g} =: \rho_t \mu_{\rm g}$. Furthermore, let $h_t := -\displaystyle{\Delta_{\mu_t}^{-1} \frac{d(b \mu)}{d\mu_t}}$, so that $\Phi^{\rm O}_{\mu_t}(b \mu) = h_t \mu_t$ and $h_0 = -\Delta_\mu^{-1} b$. Then by (\ref{define-K2}),
\begin{eqnarray} \label{KO1}
\Phi^{\rm O}_\mu {\mathcal K}^{\rm O}_\mu (a \mu, b \mu) \; & = & \; a h_0 \; \mu - \left. \frac d{dt} \right|_{t=0} (h_t \mu_t) = - \left( \left. \frac d{dt} \right|_{t=0} h_t \right)\; \mu.
\end{eqnarray}
 Now
\begin{eqnarray*}
- \frac b{1 + t a} \; & = & \; - \frac{d(b \mu)}{d\mu_t} \; = \;  \Delta_{\mu_t} h_t \; \stackrel{(\ref{Delta-mu})} = \; \Delta_{\mu_{\rm g}} h_t + \left\langle {\rm grad}\; h_t, {\rm grad}\; (\ln \rho_t ) \right\rangle\\
& = & \Delta_{\mu_{\rm g}} h_t + \left\langle {\rm grad}\; h_t, {\rm grad}\; (\ln \rho ) \right\rangle + \left\langle {\rm grad}\; h_t, {\rm grad}\; (\ln (1+ta) ) \right\rangle\\
& \stackrel{(\ref{Delta-mu})} = & \;  \Delta_\mu h_t + \langle {\rm grad}\; h_t, {\rm grad}\; (\ln (1+ta) ) \rangle.
\end{eqnarray*}
Evaluating the derivative $\frac d{dt}$ of this equation at $t = 0$ yields
\[
ab = \Delta_\mu \left. \frac d{dt} \right|_{t=0} h_t + \langle {\rm grad}\; h_0, {\rm grad}\; a \rangle = \Delta_\mu\left. \frac d{dt} \right|_{t=0} h_t - \langle {\rm grad}\; (\Delta_\mu^{-1} b), {\rm grad}\; a \rangle.
\]
Thus,
\[
\left. \frac d{dt} \right|_{t=0} h_t  = \Delta_\mu^{-1}\big(ab + \langle {\rm grad}\; (\Delta_\mu^{-1} b), {\rm grad}\; a \rangle\big),
\]
and substituting this into (\ref{KO1}) and using (\ref{Phi-O}) implies
\begin{equation} \label{KO-prelim}
{\mathcal K}^{\rm O}_\mu(a\mu, b\mu) = \big(ab + \langle {\rm grad}\; (\Delta_\mu^{-1} b), {\rm grad}\; a \rangle\big)\; \mu,
\end{equation}
and using (\ref{K-FR}), this shows (\ref{KO}). Thus, (\ref{TorO}) now follows from (\ref{Tor-alpha}). Furthermore,
\begin{eqnarray*}
{\mathcal A}^{\rm O}(a\mu, b\mu, c\mu) \; & \stackrel{(\ref{define-K})}= & \; \left\llangle c, {\mathcal K}^{\rm O}_\mu(a\mu, b\mu) \right\rrangle^{\rm O}_\mu \; \stackrel{(\ref{Otto})} = \; \left\llangle (-\Delta_\mu^{-1} c), {\mathcal K}^{\rm O}_\mu(a\mu, b\mu) \right\rrangle^{\rm FR}_\mu\\
& \stackrel{(\ref{KO-prelim})}= \; & \int_M \left(-ab (\Delta_\mu^{-1} c) - \langle (\Delta_\mu^{-1} c)\; {\rm grad}\; (\Delta_\mu^{-1} b), {\rm grad}\; a \rangle \right)\; d\mu\\
& \stackrel{(\ref{divergence})}= \; & \int_M \left(-ab (\Delta_\mu^{-1} c) + {\rm div} ((\Delta_\mu^{-1} c)\; {\rm grad}\; (\Delta_\mu^{-1} b))\; a \right)\;  d\mu\\
& \stackrel{(\ast)} = & \; \int_M a \; \langle {\rm grad} (\Delta^{-1}_\mu c), {\rm grad}(\Delta_\mu^{-1} b) \rangle \; d\mu\\
& = & \; \left \llangle a \mu, \; \langle {\rm grad} (\Delta^{-1}_\mu b), {\rm grad}(\Delta_\mu^{-1} c) \rangle \; \mu \right\rrangle^{\rm FR}_\mu,
\end{eqnarray*}
implying (\ref{AO}). Here, at $(\ast)$ we used the identity ${\rm div}(fX) = \langle {\rm grad} f, X \rangle + f {\rm div X}$.
Thus, (\ref{DG-tensor}) implies that
\[
{\mathcal D}^{\rm O}(a\mu, b\mu) = -\Delta_\mu \langle {\rm grad} (\Delta^{-1}_\mu a), {\rm grad}(\Delta_\mu^{-1} b) \rangle \; \mu,
\]
so that (\ref{LC-O}) now follows immediately from (\ref{LC}).
\end{proof}

\begin{remark}
In \cite{Ay}, the conjugate connection $\nabla^{({\rm O}, 1)}$ of $\nabla^{({\rm m})}$ w.r.t. the Otto metric was denoted by $\nabla^{{(\rm e}_1)}$. By our discussion, $\nabla^{{(\rm e}_1)}$ has vanishing curvature, but non-vanishing torsion. Moreover, $\nabla^{({\rm O}, 0)}$ and the Levi-Civita connection $\nabla^{\rm LC}$ of the Otto metric have different geodesics by Corollary \ref{Geodesics}.
\end{remark}

\subsubsection{Acknowledements} The authors have no competing interests to declare that are
relevant to the content of this article.

%
%
%
%

\end{document}